\PassOptionsToPackage{table,xcdraw}{xcolor}
\documentclass[sigconf,authorversion]{acmart}

\AtBeginDocument{%
  \providecommand\BibTeX{{%
    \normalfont B\kern-0.5em{\scshape i\kern-0.25em b}\kern-0.8em\TeX}}}

\usepackage{amsthm}
\usepackage{bm}
\usepackage[table,xcdraw]{xcolor}
\usepackage{multicol}
\usepackage{colortbl}
\usepackage{multirow}
\usepackage[normalem]{ulem}
\useunder{\uline}{\ul}{}
\usepackage{hyperref}
\usepackage{graphicx}
\usepackage{subfigure}
\usepackage{enumitem}
\usepackage{balance}

\copyrightyear{2021} 
\acmYear{2021} 
\setcopyright{acmlicensed}\acmConference[SIGIR '21]{Proceedings of the 44th International ACM SIGIR Conference on Research and Development in Information Retrieval}{July 11--15, 2021}{Virtual Event, Canada}
\acmBooktitle{Proceedings of the 44th International ACM SIGIR Conference on Research and Development in Information Retrieval (SIGIR '21), July 11--15, 2021, Virtual Event, Canada}
\acmPrice{15.00}
\acmDOI{10.1145/3404835.3463032}
\acmISBN{978-1-4503-8037-9/21/07}

\begin{document}
\fancyhead{}

\title{Cross-Batch Negative Sampling for Training \\ Two-Tower Recommenders}


\author{Jinpeng Wang}
\email{wjp20@mails.tsinghua.edu.cn}
\affiliation{
  \institution{Tsinghua University}
  \country{China}
}

\author{Jieming Zhu}
\email{jamie.zhu@huawei.com}
\affiliation{
  \institution{Huawei Noah's Ark Lab}
  \country{China}
}

\author{Xiuqiang He}
\email{hexiuqiang1@huawei.com}
\affiliation{
  \institution{Huawei Noah's Ark Lab}
  \country{China}
}

\renewcommand{\shortauthors}{Jinpeng Wang, Jieming Zhu and Xiuqiang He}

\newcommand{\ie}{\emph{i.e.},~}
\newcommand{\eg}{\emph{e.g.}~}
\newcommand{\wrt}{\emph{w.r.t.}~}
\renewcommand{\paragraph}[1]{\medskip\noindent\textbf{#1.~}}
\newcommand{\todo}{\color{red}{\textbf{TODO.}}~}

\newcommand{\xiaok}[1]{\left(#1\right)}
\newcommand{\zhongk}[1]{\left[#1\right]}
\newcommand{\dak}[1]{\left\{#1\right\}}
\newcommand{\jiaok}[1]{\left<#1\right>}
\newcommand{\shuk}[1]{\left\lVert#1\right\rVert}
\newcommand{\shuks}[1]{\left\lVert#1\right\rVert^2}
\newcommand{\shangk}[1]{\left\lceil #1 \right\rceil}
\newcommand{\xiak}[1]{\left\lfloor #1 \right\rfloor}

\newcommand{\argmax}[1]{{\mathop{\arg\mathrm{max}}_{#1}\,}}
\newcommand{\argmin}[1]{{\mathop{\arg\mathrm{min}}_{#1}\,}}
\newcommand{\pfrac}[2]{\frac{\partial #1}{\partial #2}}
\newcommand{\prob}[2]{p\xiaok{#1 \mid #2}}

\newcommand{\T}{\top}
\newcommand{\dif}{\mathop{}\!\mathrm{d}}
\newcommand{\biset}[1]{\{0,1\}^{#1}}
\newcommand{\ReLU}{\mathrm{ReLU}}
\newcommand{\relu}{\mathrm{ReLU}}
\newcommand{\sign}{\mathrm{sign}}
\newcommand\softmax{\mathrm{softmax}}
\newcommand\KL{D_{\mathrm{KL}}}
\newcommand\Var{\mathrm{Var}}
\newcommand\Cov{\mathrm{Cov}}
\newcommand{\Tr}{\mathrm{Tr}}
\newcommand{\tr}{\mathrm{tr}}
\newcommand{\dist}{\mathrm{dist}}
\newcommand{\concat}{\mathrm{concat}}
\newcommand{\mean}{\mathrm{mean}}
\newcommand{\diag}{\mathrm{diag}}
\newcommand{\cov}{\mathrm{cov}}

\newcommand{\range}[1]{{0,1,\cdots,#1}} 
\newcommand{\Range}[1]{{1,2,\cdots,#1}} 
\newcommand{\opseq}[3]{{#1_1 #3 #1_2 #3 \cdots #3 #1_{#2}}}
\newcommand{\seq}[2]{\opseq{#1}{#2}{,}}
\newcommand{\xseq}[2]{\opseq{#1}{#2}{\times}}

\newcommand{\bma}{\bm{a}}
\newcommand{\bmb}{\bm{b}}
\newcommand{\bmc}{\bm{c}}
\newcommand{\bmd}{\bm{d}}
\newcommand{\bme}{\bm{e}}
\newcommand{\bmf}{\bm{f}}
\newcommand{\bmg}{\bm{g}}
\newcommand{\bmh}{\bm{h}}
\newcommand{\bmi}{\bm{i}}
\newcommand{\bmj}{\bm{j}}
\newcommand{\bmk}{\bm{k}}
\newcommand{\bml}{\bm{l}}
\newcommand{\bmm}{\bm{m}}
\newcommand{\bmn}{\bm{n}}
\newcommand{\bmo}{\bm{o}}
\newcommand{\bmp}{\bm{p}}
\newcommand{\bmq}{\bm{q}}
\newcommand{\bmr}{\bm{r}}
\newcommand{\bms}{\bm{s}}
\newcommand{\bmt}{\bm{t}}
\newcommand{\bmu}{\bm{u}}
\newcommand{\bmv}{\bm{v}}
\newcommand{\bmw}{\bm{w}}
\newcommand{\bmx}{\bm{x}}
\newcommand{\bmy}{\bm{y}}
\newcommand{\bmz}{\bm{z}}
\newcommand{\bmzero}{\bm{0}}
\newcommand{\bmone}{\bm{1}}
\newcommand{\bmalpha}{\bm{\alpha}}
\newcommand{\bmbeta}{\bm{\beta}}
\newcommand{\bmgamma}{\bm{\gamma}}
\newcommand{\bmdelta}{\bm{\delta}}
\newcommand{\bmepsilon}{\bm{\epsilon}}
\newcommand{\bmtheta}{\bm{\theta}}
\newcommand{\bmiota}{\bm{\iota}}
\newcommand{\bmkappa}{\bm{\kappa}}
\newcommand{\bmlambda}{\bm{\lambda}}
\newcommand{\bmmu}{\bm{\mu}}
\newcommand{\bmnu}{\bm{\nu}}
\newcommand{\bmxi}{\bm{\xi}}
\newcommand{\bmpi}{\bm{\pi}}
\newcommand{\bmrho}{\bm{\rho}}
\newcommand{\bmsigma}{\bm{\sigma}}
\newcommand{\bmtau}{\bm{\tau}}
\newcommand{\bmupsilon}{\bm{\upsilon}}
\newcommand{\bmphi}{\bm{\phi}}
\newcommand{\bmchi}{\bm{\chi}}
\newcommand{\bmpsi}{\bm{\psi}}
\newcommand{\bmomega}{\bm{\omega}}
\newcommand{\bmA}{\bm{A}}
\newcommand{\bmB}{\bm{B}}
\newcommand{\bmC}{\bm{C}}
\newcommand{\bmD}{\bm{D}}
\newcommand{\bmE}{\bm{E}}
\newcommand{\bmF}{\bm{F}}
\newcommand{\bmG}{\bm{G}}
\newcommand{\bmH}{\bm{H}}
\newcommand{\bmI}{\bm{I}}
\newcommand{\bmJ}{\bm{J}}
\newcommand{\bmK}{\bm{K}}
\newcommand{\bmL}{\bm{L}}
\newcommand{\bmM}{\bm{M}}
\newcommand{\bmN}{\bm{N}}
\newcommand{\bmO}{\bm{O}}
\newcommand{\bmP}{\bm{P}}
\newcommand{\bmQ}{\bm{Q}}
\newcommand{\bmR}{\bm{R}}
\newcommand{\bmS}{\bm{S}}
\newcommand{\bmT}{\bm{T}}
\newcommand{\bmU}{\bm{U}}
\newcommand{\bmV}{\bm{V}}
\newcommand{\bmW}{\bm{W}}
\newcommand{\bmX}{\bm{X}}
\newcommand{\bmY}{\bm{Y}}
\newcommand{\bmZ}{\bm{Z}}
\newcommand{\bmGamma}{\bm{\Gamma}}
\newcommand{\bmDelta}{\bm{\Delta}}
\newcommand{\bmTheta}{\bm{\Theta}}
\newcommand{\bmLambda}{\bm{\Lambda}}
\newcommand{\bmXi}{\bm{\Xi}}
\newcommand{\bmPi}{\bm{\Pi}}
\newcommand{\bmSigma}{\bm{\Sigma}}
\newcommand{\bmUpsilon}{\bm{\Upsilon}}
\newcommand{\bmPhi}{\bm{\Phi}}
\newcommand{\bmPsi}{\bm{\Psi}}
\newcommand{\bmOmega}{\bm{\Omega}}

\newcommand{\calA}{\mathcal{A}}
\newcommand{\calB}{\mathcal{B}}
\newcommand{\calC}{\mathcal{C}}
\newcommand{\calD}{\mathcal{D}}
\newcommand{\calE}{\mathcal{E}}
\newcommand{\calF}{\mathcal{F}}
\newcommand{\calG}{\mathcal{G}}
\newcommand{\calH}{\mathcal{H}}
\newcommand{\calI}{\mathcal{I}}
\newcommand{\calJ}{\mathcal{J}}
\newcommand{\calK}{\mathcal{K}}
\newcommand{\calL}{\mathcal{L}}
\newcommand{\calM}{\mathcal{M}}
\newcommand{\calN}{\mathcal{N}}
\newcommand{\calO}{\mathcal{O}}
\newcommand{\calP}{\mathcal{P}}
\newcommand{\calQ}{\mathcal{Q}}
\newcommand{\calR}{\mathcal{R}}
\newcommand{\calS}{\mathcal{S}}
\newcommand{\calT}{\mathcal{T}}
\newcommand{\calU}{\mathcal{U}}
\newcommand{\calV}{\mathcal{V}}
\newcommand{\calW}{\mathcal{W}}
\newcommand{\calX}{\mathcal{X}}
\newcommand{\calY}{\mathcal{Y}}
\newcommand{\calZ}{\mathcal{Z}}

\newcommand{\bbC}{\mathbb{C}}
\newcommand{\bbE}{\mathbb{E}}
\newcommand{\bbN}{\mathbb{N}}
\newcommand{\bbQ}{\mathbb{Q}}
\newcommand{\bbR}{\mathbb{R}}
\newcommand{\bbZ}{\mathbb{Z}}

\newcommand{\tabincell}[2]{\begin{tabular}{@{}#1@{}}#2\end{tabular}}

\def \CBNS {CBNS}
\begin{abstract}
  The two-tower architecture has been widely applied for learning item and user representations, which is important for large-scale recommender systems.
  Many two-tower models are trained using various in-batch negative sampling strategies, where the effects of such strategies inherently rely on the size of mini-batches. 
  However, training two-tower models with a large batch size is inefficient, as it demands a large volume of memory for item and user contents and consumes a lot of time for feature encoding. 
  Interestingly, we find that neural encoders can output relatively stable features for the same input after warming up in the training process. 
  Based on such facts, we propose a simple yet effective sampling strategy called \textbf{C}ross-\textbf{B}atch \textbf{N}egative \textbf{S}ampling (\textbf{CBNS}), which takes advantage of the encoded item embeddings from recent mini-batches to boost the model training.
  Both theoretical analysis and empirical evaluations demonstrate the effectiveness and the efficiency of CBNS.
\end{abstract}




\maketitle

\section{Introduction}
The recommender systems play a key role in recommending items of personal interest to users from a vast candidate pool, where learning high-quality representations for users and items is the core challenge. 
With the rapid development of deep learning, many approaches~\cite{chen2017attentive,wei2019mmgcn,ge2020graph} are proposed to incorporate side information (\eg the content features) into modeling user-item interactions, which effectively alleviate the cold-start problem~\cite{schein2002methods} and improve the generalization ability to new items. 
Recently, neural recommenders with two-tower architecture, \ie dual neural encoder for separating item and user representations, have been shown to outperform traditional matrix factorization (MF)-based methods~\cite{hu2008collaborative, FM, SVDFeature} and become popular in large-scale systems~\cite{youtubeDNN, MIND, comirec}.

Training a recommender over the large-scale item corpus is typically translated into an extreme multi-class classification task using sampled softmax loss~\cite{IS,mikolov2013distributed,youtubeDNN}, in which the negative sampling counts for much. 
In the batch training for two-tower models, using in-batch negatives~\cite{yih2011learning,gillick2019learning}, \ie taking positive items of other users in the same mini-batch as negative items, has become a general recipe to save the computational cost of user and item encoders and improve training efficiency.
There have been some recent efforts to bridge the sampling gaps between local (\ie in-batch) and global (\ie the whole data space) distributions~\cite{corr-sfx,MNS}, and it is not new to adopt some off-the-shelf hard-mining techniques~\cite{wu2017sampling,samwalker,denseQA} to make better use of informative negatives in a batch.

\begin{figure}[t]
  \centering
  \includegraphics[width=\linewidth]{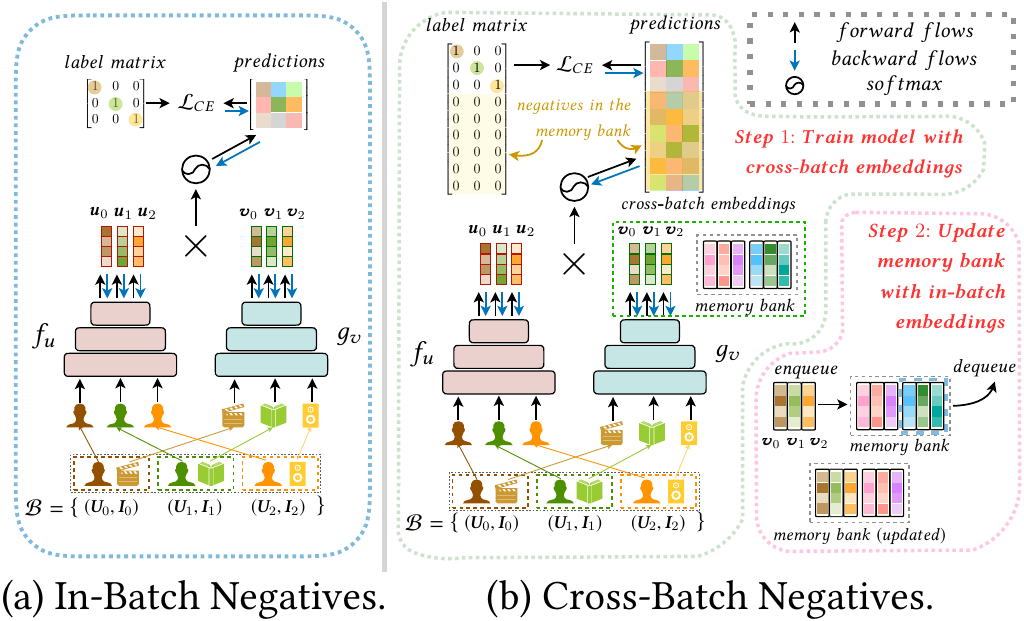}
  \caption{Sampling Strategies for Two-Tower Models}
\label{fig:method}
\end{figure}

Unfortunately, these in-batch strategies are inherently limited by the size of mini-batches. 
It is reasonable that increasing the size of mini-batches benefits negative sampling schemes and can usually boost the performance immediately, while simply enlarging a mini-batch of heavy sample contents subjects to the limited GPU memory. 
A naive solution is that, at each training iteration, we traverse the entire candidate set and collect the encoded dense representations of all candidates in a memory bank, after which we conduct negative sampling and optimize the softmax loss.
Obviously, such a solution is too time-consuming for training two-tower recommenders, but it inspires us to develop a better scheme.

In this paper, we find an interesting characteristic of the neural encoders termed \emph{embedding stability}. 
Concretely, a neural encoder tends to output stable features for the same input after the warm-up stage of the training process. 
It suggests the worth of reusing the item embeddings from recent mini-batches as the negative embeddings for current mini-batch training because the stability guarantees that the reused embeddings approximate the output of current encoder for those reused items if we encode them again. 
Motivated by such phenomenon, we propose a novel sampling approach called \textbf{C}ross-\textbf{B}atch \textbf{N}egative \textbf{S}ampling (\textbf{CBNS}). 
Specifically, we use a first-in-first-out (FIFO) memory bank for caching item embeddings across mini-batches in temporal order. 
The negatives of our CBNS come from two streams, \ie the in-batch negatives and the cached negatives from recent mini-batches in the memory bank. 
In terms of computation of encoding features, CBNS only takes the basic cost to encode the users and items in the current training mini-batch, which is as efficient as the na\"ive in-batch sampling. 
It is worth noting that CBNS effectively decouples the linear relation between the number of negatives and the size of mini-batch, thus we can leverage much more informative negatives to boost model training.
We also provide a theoretical analysis of the effectiveness of CBNS for training two-tower recommenders.

Our contributions can be summarized as follows.
\setlist{nolistsep}
\begin{itemize}[leftmargin=1.5em]
	\item[$\clubsuit$] We show the embedding stability of neural encoders, which sheds light on the motivation of reusing the encoded item embeddings from recent mini-batches.
	\item[$\clubsuit$] We propose a simple yet effective sampling approach termed \textbf{C}ross-\textbf{B}atch \textbf{N}egative \textbf{S}ampling (\textbf{CBNS}), which helps to leverage more informative negatives in training two-tower recommenders and thus boosts performance.
	\item[$\clubsuit$] We provide both theoretical analysis and empirical evaluations to demonstrate the effectiveness and the efficiency of CBNS.
\end{itemize}
\section{Related Work}
\paragraph{Two-Tower Models}
The two-tower architecture is a general framework of a query encoder along with a candidate encoder, which has been widely studied in language retrieval tasks such as question answering~\cite{denseQA}, entity retrieval~\cite{gillick2019learning} and multilingual retrieval~\cite{yang2020multilingual}.
Recently, the two-tower architecture has been adopted into large-scale recommendation~\cite{youtubeDNN,MIND,GRU4REC,comirec} and is becoming a fashion in content-aware scenarios~\cite{xu2018graphcar,ge2020graph}. 
Specially, the two-tower models in recommendation are usually applied on a much larger corpus than in language retrieval tasks, which leads to the challenge of efficient training. 

\paragraph{Negative Sampling for Two-Tower Models}
Training a retrieval model is usually reformulated as an extreme classification task, where lots of sampling-based~\cite{IS,AIS,mikolov2013distributed,blanc2018adaptive,youtubeDNN} methods have been proposed to improve the training efficiency. 
In general, these methods globally (\ie among the whole training set) sample negatives from some pre-defined distributions (\eg uniform, log-uniform or unigram distributions).
However, global sampling strategies are burdensome in training the two-tower models with content features, because: \textbf{i}) we have to make space for negatives apart from the positive pairs of users and items; \textbf{ii}) the intensive encoding becomes the bottleneck to training.
As a result, when incorporating input features into two-tower models, sharing or reusing items as negatives in a mini-batch is a typical pattern for efficient training~\cite{chen2017sampling}. 
In this direction, there are recent in-batch mechanisms for selecting hard negatives~\cite{gillick2019learning, wu2020scalable, denseQA, xiong2017end}, correcting sampling bias under an online training setting~\cite{corr-sfx} or making a mixture with global uniform sampling~\cite{MNS}. 
Note that in-batch strategies are still not satisfactory: due to the pair-wise relation of sampled users and items, the maximum number of informative negatives is essentially bounded by the batch size while enlarging the batch suffers from training inefficiency just as in global sampling. 
In contrast, we find it feasible to reuse the embedded items from the previous mini-batches
and propose a simple, efficient and effective cross-batch negative sampling scheme\footnote{A recent study in dense passage retrieval, RocketQA~\cite{RocketQA}, also introduces ``\emph{cross-batch negatives}'' that may cause name confusion with our CBNS. Different from CBNS that \emph{temporally} shares negatives across training iterations, RocketQA is trained under a distributed setting, where the negatives are only \emph{spatially} shared across the GPUs.}, which improves the aforementioned drawbacks in a low-cost way. 
We note a parallel and complementary study~\cite{zhou2020contrastive} with CBNS that focuses on reducing sampling bias in large-scale recommendation from a view of contrastive learning~\cite{moco}. It also verifies the effectiveness of memory queue in training recommenders, while we further analyse the embedding stability to justify the CBNS. 
Our work is orthogonal to existing studies of sophisticated hard negative mining~\cite{SRNS,ance} or debiasing~\cite{corr-sfx}, which can also benefit from our scheme with more negatives. 
We leave the explorations for the future work.
\section{Modeling Framework}
\subsection{Problem Formulation}
We consider the common setup for large-scale and content-aware recommendation tasks. We have two sets of user information $\calU=\dak{\bmU_i}_i^{N_\calU}$ and item information $\calI=\dak{\bmI_j}_i^{N_\calI}$ respectively, where $\bmU_i\in\calU$ amd $\bmI_j\in\calI$ are sets of pre-processed vectors of features (\eg IDs, logs and types). In a user-centric scenario, given a user with features, the goal is to retrieve a subset of items of interests. Typically, we implement it by setting up two encoders (\ie ``tower'') $f_u:\bmU\mapsto\bbR^d,\ g_v:\bmI\mapsto\bbR^d$ for users and items respectively, after which we estimate the relevance of user-item pairs by a scoring function, namely, $s(\bmU,\bmI)=f_u(\bmU)^\top g_v(\bmI)\triangleq\bmu^\top\bmv$, where $\bmu$ and $\bmv$ denote the encoded embeddings for user and item from $f_u$ and $g_v$.

\subsection{Basic Approaches}
Typically, the large-scale retrieval is treated as an extreme classification problem with a (uniformly) sampled softmax, \ie
\begin{equation}
\label{equ:uniform}
p(\bmI\mid\bmU;\bmTheta)=\frac{e^{\bmu^\top\bmv}}{e^{\bmu^\top\bmv}+\sum_{\bmI^-\in\calN}e^{\bmu^\top\bmv^-}},
\end{equation}
where $\bmTheta$ denotes the parameters of the model, $\calN$ is the sampled negative set and the superscript ``$^-$'' denotes the negatives. The models are trained with cross-entropy loss (equivalent to log-likelihood):
\begin{equation}
\calL_\text{CE}=-\frac1{|\calB|}\sum_{i\in[|\calB|]}\log p(\bmI_i\mid\bmU_i;\bmTheta).
\end{equation}

To improve the training efficiency of two-tower models, a commonly used sampling strategy is the in-batch negative sampling, as shown in Figure~\ref{fig:method}(a). Concretely, it treats other items in the same batch as negatives, and the negative distribution $q$ follows the unigram distribution based on item frequency. According to sampled softmax mechanism~\cite{IS,AIS}, we modify equation~(\ref{equ:uniform}) as
\begin{gather}
    p_\text{In-batch}(\bmI\mid\bmU;\bmTheta)=\frac{e^{s'(\bmU, \bmI; q)}}{e^{s'(\bmU, \bmI; q)}+\sum_{\bmI^-\in\calB\backslash\dak{\bmI}}e^{s'(\bmU, \bmI^-; q)}}, \\
    s'(\bmU, \bmI; q)=s(\bmU, \bmI)-\log q(\bmI)=\bmu^\top\bmv-\log q(\bmI),
\end{gather}
where $\log q(\bmI)$ is a correction to the sampling bias. In-batch negative sampling avoids extra additional negative samples to the item tower and thus saves computation cost. Unfortunately, the number of in-batch items is linearly bounded by the batch size, thus the restricted batch size on GPU limits the performance of models.

\subsection{Cross Batch Negative Sampling}

\subsubsection{Embedding Stability of Neural Model}
As the encoder keeps updating in the training, the item embeddings from past mini-batches are usually considered out-of-date and discarded. Nevertheless, we identify that such information can be reused as valid negatives in the current mini-batch, because of the \emph{embedding stability of neural model}. We investigate this phenomenon by estimating the feature drift~\cite{xbm} of the item encoder $g_v$, namely,
\begin{equation}
    D(\calI,t;\Delta t)\triangleq\sum_{\bmI\in\calI}\shuk{g_v(\bmI;\bmtheta_g^t)-g_v(\bmI;\bmtheta_g^{t-\Delta t})}_2,
\end{equation}
where $\bmtheta_g$ is the parameters of $g_v$, $t$ and $\Delta t$ denote the number and the interval of training iteration (\ie mini-batch). We train a Youtube DNN~\cite{youtubeDNN} from scratch with an in-batch negative softmax loss and compute the feature drift with different intervals in $\dak{1,5,10}$. As shown in Figure~\ref{fig:slow_drift}, the features violently change at the early stage. As the learning rate decreases, the features become relatively stable at about $4\times10^4$ iterations, making it reasonable to reuse them as valid negatives. We term such phenomenon as ``\emph{embedding stability}''. We further show in Lemma~\ref{xbm_lamma} that the embedding stability provides an upper bound for the error of gradients of the scoring function, so the stable embeddings can provide valid information for training.

\begin{lemma}
\label{xbm_lamma}
Suppose that $\shuk{\hat{\bmv}_j-\bmv_j}_2^2<\epsilon$, the output logit of scoring function is $\hat{o}_{ij}\triangleq\bmu_i^\top\hat{\bmv}_j$ and the user encoder $f_u$ satisfied Lipschitz continuous condition, then the deviation of gradient w.r.t user $\bmu_i$ is:
\begin{equation}
    \shuks{\frac{\partial\hat{o}_{ij}}{\partial\bmtheta}
    -\frac{\partial o_{ij}}{\partial\bmtheta}}_2<C\epsilon,
\end{equation}
where $C$ is the Lipschitz constant.
\end{lemma}
\begin{proof}
The approximated gradient error can be compueted as:
\begin{align}
    &\shuks{\frac{\partial\hat{o}_{ij}}{\partial\bmtheta}
    -\frac{\partial o_{ij}}{\partial\bmtheta}}_2 =\shuks{\xiaok{\frac{\partial \hat{o}_{ij}}{\partial\bmu_i}-\frac{\partial o_{ij}}{\partial\bmu_i}}\frac{\partial \bmu_i}{\partial\bmtheta}}_2 
    =\shuks{(\hat{\bmv}_j-\bmv_j)\frac{\partial\bmu_i}{\partial\bmtheta}}_2 \\
    \le&\shuks{\hat{\bmv}_j-\bmv_j}_2\shuks{\frac{\partial\bmu_i}{\partial\bmtheta}}_2 
    \le\shuks{\frac{\partial f_u(\bmU_i;\bmtheta^t)}{\partial\bmtheta}}_2\epsilon\le C\epsilon.
\end{align}
Empirically, $C\le1$ holds with the models used in our experiments, thus the gradient error can be controlled by embedding stability.
\end{proof}

\subsubsection{FIFO Memory Bank for Cross Batch Features}
Since the embeddings change relatively violently at the early stage, we warm up the item encoder with na\"ive in-batch negative samplings for $4\times10^{4}$ iterations, which helps the model approximate a local optimal and produce stable embeddings. Then we begin to train recommenders with an FIFO memory bank $\calM=\dak{\xiaok{\bmv_i, q(\bmI_i)}}_{i=1}^M$, where $q(\bmI_i)$ denotes the sampling probability of item $\bmI_i$ under the unigram distribution $q$ and $M$ is the memory size. 
The cross-batch negative sampling (CBNS) with the FIFO memory bank is illustrated in Figure~\ref{fig:method}(b), and the output of softmax of CBNS is formulated as 
\begin{equation}
    p_\text{CBNS}(\bmI\mid\bmU;\bmTheta)=\frac{e^{s'(\bmU, \bmI; q)}}{e^{s'(\bmU, \bmI; q)}+\sum_{\bmI^-\in\calM\cup\calB\backslash\dak{\bmI}}e^{s'(\bmU, \bmI^-; q)}},
\end{equation}
At the end of each iteration, we enqueue the embeddings and the respective sampling probabilities in the current mini-batch and dequeue the earliest ones. Note that our memory bank is updated with embeddings without any additional computation. 
Besides, the size of the memory bank can be relatively large, because it does not require much memory cost for those embeddings.

\begin{figure}[t]
  \centering
  \includegraphics[width=0.5\linewidth]{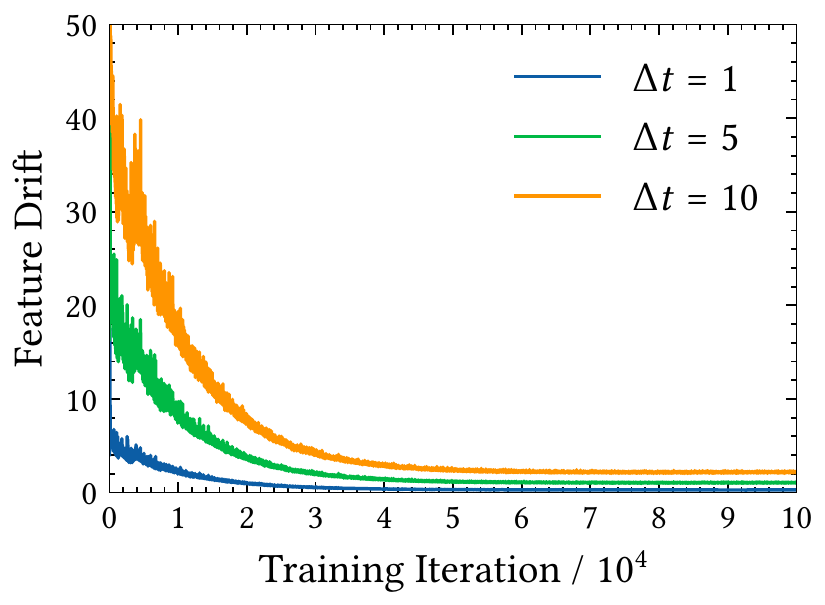}
  \caption{Feature drift of YoutubeDNN~\cite{youtubeDNN} \wrt $\Delta t$s on Amazon-Books dataset~\cite{10.1145/2872427.2883037,10.1145/2766462.2767755}.}
\label{fig:slow_drift}
\end{figure}
\section{Experiments and Results}
In this session, we conduct experiments with 3 typical two-tower models and compare CBNS with 3 representative competitors. 
Besides, we explore the choices of different memory bank sizes.

\subsection{Experimental Setup}
\paragraph{Dataset}
We conduct experiments on a challenging public dataset, Amazon\footnote{http://jmcauley.ucsd.edu/data/amazon/}, which consists of product reviews and metadata from Amazon~\cite{10.1145/2766462.2767755,10.1145/2872427.2883037}. In our experiment, we use the Books category of the Amazon dataset. The Amazon-book dataset is a collection of 8,898,041 user-item interactions between 459,133 users and 313,966 items. To standardize the evaluations, we follow~\cite{liang2018variational,comirec} to split all users into training/validation/test sets by 8:1:1.

\paragraph{Models}
To generalize the conclusions, we conduct experiments on 3 representative two-tower models,
\ie Youtube DNN~\cite{youtubeDNN}, GRU4REC~\cite{GRU4REC} and MIND~\cite{MIND}, with item and user features.

\paragraph{Negative Sampling Strategies}
We compare our CBNS with three representative sampling strategies: \textbf{i}) Uniform Sampling; \textbf{ii}) In-Batch Negative Sampling; \textbf{iii}) Mixed Negative Sampling~\cite{MNS} is a recent work to improve the in-batch sampling by additionally sampling some negatives from the global uniform distribution. 

\paragraph{Metrics}
We report \emph{results of convergence} for all experiments and use the following metrics in our evaluations in terms of performance and training efficiency. 
\textbf{i}) Recall. 
\textbf{ii}) Normalized Discounted Cumulative Gain (NDCG). 
\textbf{iii}) Convergence Time. We report the convergence time in minutes for different sampling strategies with different models.
\textbf{iv}) Average Training Time per 1000 Mini-Batches. To compare the batch-wise training efficiency, we also report the average time in seconds for training every 1000 Mini-Batches.

\paragraph{Parameter Configurations}
The embedding dimension $d$ is set to 64. 
The size of mini-batch is set to 128.  
The coefficient for $\ell_2$ regularization is selected from $\dak{0,10^{-6},10^{-5},10^{-4},10^{-3}}$ via cross-validation.
We use the Adam optimizer~\cite{Adam} with the base $lr=0.001$. 
The patience counter of the early stopping is set to 20. 
The numbers of globally sampled negatives for uniform sampling and MNS are set to 1280 and 1152 (there are 128 in-batch negatives in MNS). The default memory bank size $M$ for CBNS is 2432.

\subsection{Results}

\begin{table}[!t]
\renewcommand\arraystretch{1.01}
\centering
\caption{Model performance with different sampling strategies on the Amazon-Books dataset. Bolded numbers are the best performance of each group. ``Conv. Time'' and ``Avg. Time'' denote the training time till the early stop (in minutes) and the average time per $10^3$ mini-batches (in seconds).}

\resizebox{0.9\columnwidth}{!}{
\begin{tabular}{c|c|c|c|c|c|c|c}
\hline
\multirow{2}{*}{Model} & \multirow{2}{*}{\begin{tabular}[c]{@{}c@{}}Negative\\Sampling\end{tabular}} & \multirow{2}{*}{\begin{tabular}[c]{@{}c@{}}Avg.\\Time\end{tabular}} & \multirow{2}{*}{\begin{tabular}[c]{@{}c@{}}Conv.\\Time\end{tabular}} & \multicolumn{2}{l|}{Metrics@20(\%)} & \multicolumn{2}{l}{Metrics@50(\%)} \\ \cline{5-8} 
 &  &  &  & Recall & NDCG & Recall & NDCG \\ \hline
\multirow{4}{*}{\begin{tabular}[c]{@{}c@{}}YouTube\\DNN\end{tabular}} & Uniform & 26.57 & 135.97 & 4.205 & 7.174 & 7.028 & 11.741 \\ 
 & In-Batch & 16.27 & \phantom{0}62.91 & 4.070 & 6.962 & 6.908 & 11.509 \\ 
 & MNS & 24.73 & 138.50 & 4.113 & 7.050 & 6.929 & 11.535 \\ 
 & {\cellcolor[HTML]{EAFAF1}CBNS} & {\cellcolor[HTML]{EAFAF1}20.38} & {\cellcolor[HTML]{EAFAF1}\phantom{0}62.50} & {\cellcolor[HTML]{EAFAF1}\textbf{4.237}} & {\cellcolor[HTML]{EAFAF1}\textbf{7.325}} & {\cellcolor[HTML]{EAFAF1}\textbf{7.084}} & {\cellcolor[HTML]{EAFAF1}\textbf{11.881}} \\ \hline
\multirow{4}{*}{GRU4REC} & Uniform & 96.08 & 326.68 & 4.237 & 7.360 & 6.895 & 11.554 \\ 
 & In-Batch & 35.09 & \phantom{0}91.81 & 3.912 & 6.674 & 6.425 & 10.724 \\ 
 & MNS & 58.69 & 464.62 & 4.110 & 7.178 & 6.553 & 11.183 \\ 
 & {\cellcolor[HTML]{EAFAF1}CBNS} & {\cellcolor[HTML]{EAFAF1}40.15} & {\cellcolor[HTML]{EAFAF1}167.31} & {\cellcolor[HTML]{EAFAF1}\textbf{4.257}} & {\cellcolor[HTML]{EAFAF1}\textbf{7.479}} & {\cellcolor[HTML]{EAFAF1}\textbf{7.044}} & {\cellcolor[HTML]{EAFAF1}\textbf{12.032}} \\ \hline
\multirow{4}{*}{MIND} & Uniform & 37.51 & 180.66 & 4.137 & 6.881 & 7.008 & 11.400 \\ 
 & In-Batch & 27.87 & 104.50 & 4.140 & 7.009 & 6.894 & 11.363 \\ 
 & MNS & 36.82 & 205.57 & 4.576 & 7.596 & 7.135 & 11.603 \\ 
 & {\cellcolor[HTML]{EAFAF1}CBNS} & {\cellcolor[HTML]{EAFAF1}31.65} & {\cellcolor[HTML]{EAFAF1}146.14} & {\cellcolor[HTML]{EAFAF1}\textbf{4.799}} & {\cellcolor[HTML]{EAFAF1}\textbf{8.105}} & {\cellcolor[HTML]{EAFAF1}\textbf{7.689}} & {\cellcolor[HTML]{EAFAF1}\textbf{12.671}} \\ \hline
\end{tabular}
}
\label{tbl:bookresults}
\end{table}


\paragraph{Effectiveness of CBNS}
We first show the results of the combinations of 3 models and 4 sampling strategies in Table~\ref{tbl:bookresults}. Our CBNS consistently keeps the superior on Recalls and NDCGs over other sampling strategies with all tested models. For CBNS, it takes an acceptably longer average time per 1000 mini-batches than the in-batch sampling to compute similarities of cached embeddings, while it gets \textbf{3.23\%}, \textbf{12.19\%}, \textbf{11.51\%} improvements to in-batch sampling with all tested models on NDCG@50, because of involving more informative negatives in batch training.
We also discover several interesting insights from Table~\ref{tbl:bookresults}: \textbf{1}) Sharing encoded item embeddings in batch or reusing item embeddings across batches can improve the average training time, while it doesn't mean accelerating convergence. The MNS converges much slower than other competitors since it takes more time for models to adapt to the mixed distribution. \textbf{2}) Though in-batch sampling is the most efficient strategy when we keep the same batch size for all strategies, we find it inferior because of the lack of informative negatives. \textbf{3}) Generally, the in-batch sampling and CBNS are unigram distributed, which is more similar to the true negative distribution and makes it compatible for models to serve online. The validation recall and NDCG curves \wrt the training iteration and the wall time respectively are shown in Figure~\ref{fig:res_vis}, which demonstrates the fast convergence and satisfactory performance of CBNS.

\begin{figure}[t]
  \centering
  \includegraphics[width=\columnwidth]{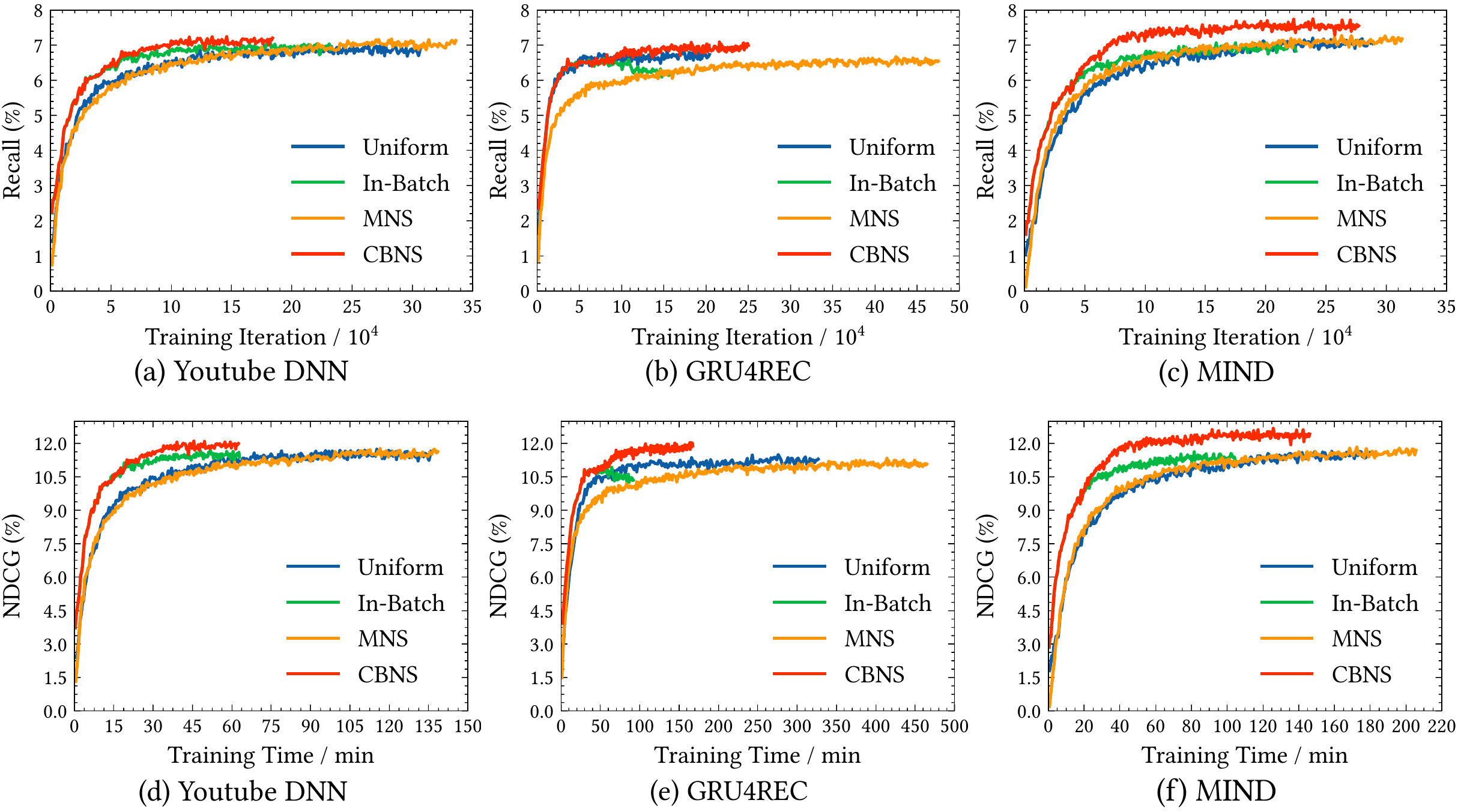}
  \caption{Validation recall curves (a-c) \wrt the training iteration (\ie mini-batch) and validation NDCG curves (d-f) \wrt the training wall time, all under the same batch size.}
\label{fig:res_vis}
\end{figure}

\begin{table}[t]
\renewcommand\arraystretch{1.01}
\centering
\caption{
Metrics@50 (\%) of 3 models trained with CBNS of different size of memory bank. Bolded numbers are thebest performance of each column. All the batch sizes are 128.}
\resizebox{0.9\columnwidth}{!}{
\begin{tabular}{c|c|c|c|c|c|c|c}
\hline
\multirow{2}{*}{$M$} & \multirow{2}{*}{\#neg} & \multicolumn{2}{c|}{Youtube DNN} & \multicolumn{2}{c|}{GRU4REC} & \multicolumn{2}{c}{MIND} \\ \cline{3-8} 
 &  & Recall & NDCG & Recall & NDCG & Recall & NDCG \\ \hline
0 & 128 & 6.908 & 11.509 & 6.425 & 10.724 & 6.894 & 11.363 \\ 
128 & 256 & 7.032 & 11.877 & 6.518 & 10.993 & 7.003 & 11.431 \\
512 & 640 & 7.175 & \textbf{12.085} & 6.679 & 11.229 & 7.102 & 11.66 \\
1152 & 1280 & \textbf{7.182} & 12.055 & 6.845 & 11.681 & 7.463 & 12.309 \\
2432 & 2560 & 7.084 & 11.881 & 7.044 & 12.032 & \textbf{7.689} & \textbf{12.671} \\
4992 & 5120 & 6.988 & 11.807 & \textbf{7.394} & \textbf{12.423} & 7.474 & 12.303 \\
10112 & 10240 & 7.001 & 11.772 & 6.903 & 11.810 & 7.268 & 11.939 \\ \hline
\end{tabular}}
\label{tbl:ablation}
\end{table}

\paragraph{Memory Bank Size $M$ for CBNS}
We further explore CBNS by investigating 
how the model performance changes under different memory bank size $M$ with a fixed mini-batch size, as shown in Table~\ref{tbl:ablation}. 
Different models have different best $M$s. The Youtube DNN performs almost best when $M$ is 512 or 1152, while GRU4REC and MIND perform best under larger $M$s as $4992$ and $2432$, respectively. 
Note that embedding stability only holds within a limited period, thus an overlarge $M$ brings less benefit or adversely hurts models.
\section{Conclusions}

In this paper, we propose a novel Cross-Batch Negative Sampling (CBNS) strategy for efficiently training two-tower recommenders in a content-aware scenario, which breaks the dilemma of deciding between effectiveness and efficiency in existing sampling schemes. 
We find that neural encoders have embedding stability in the training, which enables reusing cross-batch negative item embeddings to boost training. 
Base on such facts, we set up a memory bank for caching the item embeddings from previous iterations and conduct cross-batch negative sampling for training. 
Compared with commonly used sampling strategies across three popular two-tower models, CBNS can efficiently involve more informative negatives, thus boosting performance with a low cost.
Future work includes exploring the combinations of CBNS with negative mining techniques to further improve the training of the two-tower recommenders.

\bibliographystyle{ACM-Reference-Format}

\balance
\bibliography{references}


\begin{thebibliography}{37}


\ifx \showCODEN    \undefined \def \showCODEN     #1{\unskip}     \fi
\ifx \showDOI      \undefined \def \showDOI       #1{#1}\fi
\ifx \showISBNx    \undefined \def \showISBNx     #1{\unskip}     \fi
\ifx \showISBNxiii \undefined \def \showISBNxiii  #1{\unskip}     \fi
\ifx \showISSN     \undefined \def \showISSN      #1{\unskip}     \fi
\ifx \showLCCN     \undefined \def \showLCCN      #1{\unskip}     \fi
\ifx \shownote     \undefined \def \shownote      #1{#1}          \fi
\ifx \showarticletitle \undefined \def \showarticletitle #1{#1}   \fi
\ifx \showURL      \undefined \def \showURL       {\relax}        \fi
\providecommand\bibfield[2]{#2}
\providecommand\bibinfo[2]{#2}
\providecommand\natexlab[1]{#1}
\providecommand\showeprint[2][]{arXiv:#2}

\bibitem[\protect\citeauthoryear{Bengio and Sen{\'e}cal}{Bengio and
  Sen{\'e}cal}{2008}]%
        {AIS}
\bibfield{author}{\bibinfo{person}{Yoshua Bengio} {and}
  \bibinfo{person}{Jean-S{\'e}bastien Sen{\'e}cal}.}
  \bibinfo{year}{2008}\natexlab{}.
\newblock \showarticletitle{Adaptive importance sampling to accelerate training
  of a neural probabilistic language model}.
\newblock \bibinfo{journal}{\emph{IEEE Transactions on Neural Networks}}
  \bibinfo{volume}{19}, \bibinfo{number}{4} (\bibinfo{year}{2008}),
  \bibinfo{pages}{713--722}.
\newblock


\bibitem[\protect\citeauthoryear{Bengio, Sen{\'e}cal, et~al\mbox{.}}{Bengio
  et~al\mbox{.}}{2003}]%
        {IS}
\bibfield{author}{\bibinfo{person}{Yoshua Bengio},
  \bibinfo{person}{Jean-S{\'e}bastien Sen{\'e}cal}, {et~al\mbox{.}}}
  \bibinfo{year}{2003}\natexlab{}.
\newblock \showarticletitle{Quick Training of Probabilistic Neural Nets by
  Importance Sampling.}. In \bibinfo{booktitle}{\emph{International Conference
  on Artificial Intelligence and Statistics (AISTATS)}}. PMLR,
  \bibinfo{pages}{1--9}.
\newblock


\bibitem[\protect\citeauthoryear{Blanc and Rendle}{Blanc and Rendle}{2018}]%
        {blanc2018adaptive}
\bibfield{author}{\bibinfo{person}{Guy Blanc} {and} \bibinfo{person}{Steffen
  Rendle}.} \bibinfo{year}{2018}\natexlab{}.
\newblock \showarticletitle{Adaptive sampled softmax with kernel based
  sampling}. In \bibinfo{booktitle}{\emph{International Conference on Machine
  Learning}}. PMLR, \bibinfo{pages}{590--599}.
\newblock


\bibitem[\protect\citeauthoryear{Cen, Zhang, Zou, Zhou, Yang, and Tang}{Cen
  et~al\mbox{.}}{2020}]%
        {comirec}
\bibfield{author}{\bibinfo{person}{Yukuo Cen}, \bibinfo{person}{Jianwei Zhang},
  \bibinfo{person}{Xu Zou}, \bibinfo{person}{Chang Zhou},
  \bibinfo{person}{Hongxia Yang}, {and} \bibinfo{person}{Jie Tang}.}
  \bibinfo{year}{2020}\natexlab{}.
\newblock \showarticletitle{Controllable Multi-Interest Framework for
  Recommendation}. In \bibinfo{booktitle}{\emph{Proceedings of the 26th ACM
  SIGKDD International Conference on Knowledge Discovery \& Data Mining}}.
  \bibinfo{pages}{2942--2951}.
\newblock


\bibitem[\protect\citeauthoryear{Chen, Wang, Zhou, Shi, Feng, and Chen}{Chen
  et~al\mbox{.}}{2019}]%
        {samwalker}
\bibfield{author}{\bibinfo{person}{Jiawei Chen}, \bibinfo{person}{Can Wang},
  \bibinfo{person}{Sheng Zhou}, \bibinfo{person}{Qihao Shi},
  \bibinfo{person}{Yan Feng}, {and} \bibinfo{person}{Chun Chen}.}
  \bibinfo{year}{2019}\natexlab{}.
\newblock \showarticletitle{Samwalker: Social recommendation with informative
  sampling strategy}. In \bibinfo{booktitle}{\emph{The World Wide Web
  Conference}}. \bibinfo{pages}{228--239}.
\newblock


\bibitem[\protect\citeauthoryear{Chen, Zhang, He, Nie, Liu, and Chua}{Chen
  et~al\mbox{.}}{2017b}]%
        {chen2017attentive}
\bibfield{author}{\bibinfo{person}{Jingyuan Chen}, \bibinfo{person}{Hanwang
  Zhang}, \bibinfo{person}{Xiangnan He}, \bibinfo{person}{Liqiang Nie},
  \bibinfo{person}{Wei Liu}, {and} \bibinfo{person}{Tat-Seng Chua}.}
  \bibinfo{year}{2017}\natexlab{b}.
\newblock \showarticletitle{Attentive collaborative filtering: Multimedia
  recommendation with item-and component-level attention}. In
  \bibinfo{booktitle}{\emph{Proceedings of the 40th International ACM SIGIR
  conference on Research and Development in Information Retrieval}}.
  \bibinfo{pages}{335--344}.
\newblock


\bibitem[\protect\citeauthoryear{Chen, Sun, Shi, and Hong}{Chen
  et~al\mbox{.}}{2017a}]%
        {chen2017sampling}
\bibfield{author}{\bibinfo{person}{Ting Chen}, \bibinfo{person}{Yizhou Sun},
  \bibinfo{person}{Yue Shi}, {and} \bibinfo{person}{Liangjie Hong}.}
  \bibinfo{year}{2017}\natexlab{a}.
\newblock \showarticletitle{On sampling strategies for neural network-based
  collaborative filtering}. In \bibinfo{booktitle}{\emph{Proceedings of the
  23rd ACM SIGKDD International Conference on Knowledge Discovery and Data
  Mining}}. \bibinfo{pages}{767--776}.
\newblock


\bibitem[\protect\citeauthoryear{Chen, Zhang, Lu, Chen, Zheng, and Yu}{Chen
  et~al\mbox{.}}{2012}]%
        {SVDFeature}
\bibfield{author}{\bibinfo{person}{Tianqi Chen}, \bibinfo{person}{Weinan
  Zhang}, \bibinfo{person}{Qiuxia Lu}, \bibinfo{person}{Kailong Chen},
  \bibinfo{person}{Zhao Zheng}, {and} \bibinfo{person}{Yong Yu}.}
  \bibinfo{year}{2012}\natexlab{}.
\newblock \showarticletitle{SVDFeature: a toolkit for feature-based
  collaborative filtering}.
\newblock \bibinfo{journal}{\emph{The Journal of Machine Learning Research}}
  \bibinfo{volume}{13}, \bibinfo{number}{1} (\bibinfo{year}{2012}),
  \bibinfo{pages}{3619--3622}.
\newblock


\bibitem[\protect\citeauthoryear{Covington, Adams, and Sargin}{Covington
  et~al\mbox{.}}{2016}]%
        {youtubeDNN}
\bibfield{author}{\bibinfo{person}{Paul Covington}, \bibinfo{person}{Jay
  Adams}, {and} \bibinfo{person}{Emre Sargin}.}
  \bibinfo{year}{2016}\natexlab{}.
\newblock \showarticletitle{Deep neural networks for youtube recommendations}.
  In \bibinfo{booktitle}{\emph{Proceedings of the 10th ACM conference on
  recommender systems}}. \bibinfo{pages}{191--198}.
\newblock


\bibitem[\protect\citeauthoryear{Ding, Quan, Yao, Li, and Jin}{Ding
  et~al\mbox{.}}{2020b}]%
        {SRNS}
\bibfield{author}{\bibinfo{person}{Jingtao Ding}, \bibinfo{person}{Yuhan Quan},
  \bibinfo{person}{Quanming Yao}, \bibinfo{person}{Yong Li}, {and}
  \bibinfo{person}{Depeng Jin}.} \bibinfo{year}{2020}\natexlab{b}.
\newblock \showarticletitle{Simplify and Robustify Negative Sampling for
  Implicit Collaborative Filtering}. In \bibinfo{booktitle}{\emph{Advances in
  Neural Information Processing Systems (NeurIPS)}},
  \bibfield{editor}{\bibinfo{person}{H.~Larochelle},
  \bibinfo{person}{M.~Ranzato}, \bibinfo{person}{R.~Hadsell},
  \bibinfo{person}{M.~F. Balcan}, {and} \bibinfo{person}{H.~Lin}} (Eds.),
  Vol.~\bibinfo{volume}{33}. \bibinfo{pages}{1094--1105}.
\newblock


\bibitem[\protect\citeauthoryear{Ding, Liu, Liu, Ren, Zhao, Dong, Wu, and
  Wang}{Ding et~al\mbox{.}}{2020a}]%
        {RocketQA}
\bibfield{author}{\bibinfo{person}{Yingqi Qu~Yuchen Ding},
  \bibinfo{person}{Jing Liu}, \bibinfo{person}{Kai Liu},
  \bibinfo{person}{Ruiyang Ren}, \bibinfo{person}{Xin Zhao},
  \bibinfo{person}{Daxiang Dong}, \bibinfo{person}{Hua Wu}, {and}
  \bibinfo{person}{Haifeng Wang}.} \bibinfo{year}{2020}\natexlab{a}.
\newblock \showarticletitle{RocketQA: An Optimized Training Approach to Dense
  Passage Retrieval for Open-Domain Question Answering}.
\newblock \bibinfo{journal}{\emph{arXiv preprint arXiv:2010.08191}}
  (\bibinfo{year}{2020}).
\newblock


\bibitem[\protect\citeauthoryear{Ge, Wu, Wu, Qi, and Huang}{Ge
  et~al\mbox{.}}{2020}]%
        {ge2020graph}
\bibfield{author}{\bibinfo{person}{Suyu Ge}, \bibinfo{person}{Chuhan Wu},
  \bibinfo{person}{Fangzhao Wu}, \bibinfo{person}{Tao Qi}, {and}
  \bibinfo{person}{Yongfeng Huang}.} \bibinfo{year}{2020}\natexlab{}.
\newblock \showarticletitle{Graph enhanced representation learning for news
  recommendation}. In \bibinfo{booktitle}{\emph{Proceedings of The Web
  Conference 2020}}. \bibinfo{pages}{2863--2869}.
\newblock


\bibitem[\protect\citeauthoryear{Gillick, Kulkarni, Lansing, Presta, Baldridge,
  Ie, and Garcia-Olano}{Gillick et~al\mbox{.}}{2019}]%
        {gillick2019learning}
\bibfield{author}{\bibinfo{person}{Dan Gillick}, \bibinfo{person}{Sayali
  Kulkarni}, \bibinfo{person}{Larry Lansing}, \bibinfo{person}{Alessandro
  Presta}, \bibinfo{person}{Jason Baldridge}, \bibinfo{person}{Eugene Ie},
  {and} \bibinfo{person}{Diego Garcia-Olano}.} \bibinfo{year}{2019}\natexlab{}.
\newblock \showarticletitle{Learning Dense Representations for Entity
  Retrieval}. In \bibinfo{booktitle}{\emph{Proceedings of the 23rd Conference
  on Computational Natural Language Learning (CoNLL)}}.
  \bibinfo{pages}{528--537}.
\newblock


\bibitem[\protect\citeauthoryear{He, Fan, Wu, Xie, and Girshick}{He
  et~al\mbox{.}}{2020}]%
        {moco}
\bibfield{author}{\bibinfo{person}{Kaiming He}, \bibinfo{person}{Haoqi Fan},
  \bibinfo{person}{Yuxin Wu}, \bibinfo{person}{Saining Xie}, {and}
  \bibinfo{person}{Ross Girshick}.} \bibinfo{year}{2020}\natexlab{}.
\newblock \showarticletitle{Momentum contrast for unsupervised visual
  representation learning}. In \bibinfo{booktitle}{\emph{Proceedings of the
  IEEE/CVF Conference on Computer Vision and Pattern Recognition}}.
  \bibinfo{pages}{9729--9738}.
\newblock


\bibitem[\protect\citeauthoryear{He and McAuley}{He and McAuley}{2016}]%
        {10.1145/2872427.2883037}
\bibfield{author}{\bibinfo{person}{Ruining He} {and} \bibinfo{person}{Julian
  McAuley}.} \bibinfo{year}{2016}\natexlab{}.
\newblock \showarticletitle{Ups and Downs: Modeling the Visual Evolution of
  Fashion Trends with One-Class Collaborative Filtering}. In
  \bibinfo{booktitle}{\emph{Proceedings of the 25th International Conference on
  World Wide Web}}. \bibinfo{pages}{507–517}.
\newblock


\bibitem[\protect\citeauthoryear{Hidasi, Karatzoglou, Baltrunas, and
  Tikk}{Hidasi et~al\mbox{.}}{2015}]%
        {GRU4REC}
\bibfield{author}{\bibinfo{person}{Bal{\'a}zs Hidasi},
  \bibinfo{person}{Alexandros Karatzoglou}, \bibinfo{person}{Linas Baltrunas},
  {and} \bibinfo{person}{Domonkos Tikk}.} \bibinfo{year}{2015}\natexlab{}.
\newblock \showarticletitle{Session-based recommendations with recurrent neural
  networks}.
\newblock \bibinfo{journal}{\emph{arXiv preprint arXiv:1511.06939}}
  (\bibinfo{year}{2015}).
\newblock


\bibitem[\protect\citeauthoryear{Hu, Koren, and Volinsky}{Hu
  et~al\mbox{.}}{2008}]%
        {hu2008collaborative}
\bibfield{author}{\bibinfo{person}{Yifan Hu}, \bibinfo{person}{Yehuda Koren},
  {and} \bibinfo{person}{Chris Volinsky}.} \bibinfo{year}{2008}\natexlab{}.
\newblock \showarticletitle{Collaborative filtering for implicit feedback
  datasets}. In \bibinfo{booktitle}{\emph{2008 Eighth IEEE International
  Conference on Data Mining}}. IEEE, \bibinfo{pages}{263--272}.
\newblock


\bibitem[\protect\citeauthoryear{Karpukhin, Oguz, Min, Lewis, Wu, Edunov, Chen,
  and Yih}{Karpukhin et~al\mbox{.}}{2020}]%
        {denseQA}
\bibfield{author}{\bibinfo{person}{Vladimir Karpukhin}, \bibinfo{person}{Barlas
  Oguz}, \bibinfo{person}{Sewon Min}, \bibinfo{person}{Patrick Lewis},
  \bibinfo{person}{Ledell Wu}, \bibinfo{person}{Sergey Edunov},
  \bibinfo{person}{Danqi Chen}, {and} \bibinfo{person}{Wen-tau Yih}.}
  \bibinfo{year}{2020}\natexlab{}.
\newblock \showarticletitle{Dense Passage Retrieval for Open-Domain Question
  Answering}. In \bibinfo{booktitle}{\emph{Proceedings of the 2020 Conference
  on Empirical Methods in Natural Language Processing (EMNLP)}}.
  \bibinfo{pages}{6769--6781}.
\newblock


\bibitem[\protect\citeauthoryear{Kingma and Ba}{Kingma and Ba}{2015}]%
        {Adam}
\bibfield{author}{\bibinfo{person}{Diederik~P. Kingma} {and}
  \bibinfo{person}{Jimmy Ba}.} \bibinfo{year}{2015}\natexlab{}.
\newblock \showarticletitle{Adam: A Method for Stochastic Optimization}. In
  \bibinfo{booktitle}{\emph{ICLR}}.
\newblock


\bibitem[\protect\citeauthoryear{Li, Liu, Wu, Xu, Zhao, Huang, Kang, Chen, Li,
  and Lee}{Li et~al\mbox{.}}{2019}]%
        {MIND}
\bibfield{author}{\bibinfo{person}{Chao Li}, \bibinfo{person}{Zhiyuan Liu},
  \bibinfo{person}{Mengmeng Wu}, \bibinfo{person}{Yuchi Xu},
  \bibinfo{person}{Huan Zhao}, \bibinfo{person}{Pipei Huang},
  \bibinfo{person}{Guoliang Kang}, \bibinfo{person}{Qiwei Chen},
  \bibinfo{person}{Wei Li}, {and} \bibinfo{person}{Dik~Lun Lee}.}
  \bibinfo{year}{2019}\natexlab{}.
\newblock \showarticletitle{Multi-interest network with dynamic routing for
  recommendation at Tmall}. In \bibinfo{booktitle}{\emph{Proceedings of the
  28th ACM International Conference on Information and Knowledge Management}}.
  \bibinfo{pages}{2615--2623}.
\newblock


\bibitem[\protect\citeauthoryear{Liang, Krishnan, Hoffman, and Jebara}{Liang
  et~al\mbox{.}}{2018}]%
        {liang2018variational}
\bibfield{author}{\bibinfo{person}{Dawen Liang}, \bibinfo{person}{Rahul~G
  Krishnan}, \bibinfo{person}{Matthew~D Hoffman}, {and} \bibinfo{person}{Tony
  Jebara}.} \bibinfo{year}{2018}\natexlab{}.
\newblock \showarticletitle{Variational autoencoders for collaborative
  filtering}. In \bibinfo{booktitle}{\emph{Proceedings of the 2018 world wide
  web conference}}. \bibinfo{pages}{689--698}.
\newblock


\bibitem[\protect\citeauthoryear{McAuley, Targett, Shi, and van~den
  Hengel}{McAuley et~al\mbox{.}}{2015}]%
        {10.1145/2766462.2767755}
\bibfield{author}{\bibinfo{person}{Julian McAuley},
  \bibinfo{person}{Christopher Targett}, \bibinfo{person}{Qinfeng Shi}, {and}
  \bibinfo{person}{Anton van~den Hengel}.} \bibinfo{year}{2015}\natexlab{}.
\newblock \showarticletitle{Image-Based Recommendations on Styles and
  Substitutes}. In \bibinfo{booktitle}{\emph{Proceedings of the 38th
  International ACM SIGIR Conference on Research and Development in Information
  Retrieval.}} \bibinfo{pages}{43–52}.
\newblock


\bibitem[\protect\citeauthoryear{Mikolov, Sutskever, Chen, Corrado, and
  Dean}{Mikolov et~al\mbox{.}}{2013}]%
        {mikolov2013distributed}
\bibfield{author}{\bibinfo{person}{Tomas Mikolov}, \bibinfo{person}{Ilya
  Sutskever}, \bibinfo{person}{Kai Chen}, \bibinfo{person}{Greg Corrado}, {and}
  \bibinfo{person}{Jeffrey Dean}.} \bibinfo{year}{2013}\natexlab{}.
\newblock \showarticletitle{Distributed representations of words and phrases
  and their compositionality}. In \bibinfo{booktitle}{\emph{Proceedings of the
  26th International Conference on Neural Information Processing Systems-Volume
  2}}. \bibinfo{pages}{3111--3119}.
\newblock


\bibitem[\protect\citeauthoryear{Rendle}{Rendle}{2010}]%
        {FM}
\bibfield{author}{\bibinfo{person}{Steffen Rendle}.}
  \bibinfo{year}{2010}\natexlab{}.
\newblock \showarticletitle{Factorization machines}. In
  \bibinfo{booktitle}{\emph{2010 IEEE International Conference on Data
  Mining}}. IEEE, \bibinfo{pages}{995--1000}.
\newblock


\bibitem[\protect\citeauthoryear{Schein, Popescul, Ungar, and Pennock}{Schein
  et~al\mbox{.}}{2002}]%
        {schein2002methods}
\bibfield{author}{\bibinfo{person}{Andrew~I Schein},
  \bibinfo{person}{Alexandrin Popescul}, \bibinfo{person}{Lyle~H Ungar}, {and}
  \bibinfo{person}{David~M Pennock}.} \bibinfo{year}{2002}\natexlab{}.
\newblock \showarticletitle{Methods and metrics for cold-start
  recommendations}. In \bibinfo{booktitle}{\emph{Proceedings of the 25th annual
  international ACM SIGIR conference on Research and development in information
  retrieval}}. \bibinfo{pages}{253--260}.
\newblock


\bibitem[\protect\citeauthoryear{Wang, Zhang, Huang, and Scott}{Wang
  et~al\mbox{.}}{2020}]%
        {xbm}
\bibfield{author}{\bibinfo{person}{Xun Wang}, \bibinfo{person}{Haozhi Zhang},
  \bibinfo{person}{Weilin Huang}, {and} \bibinfo{person}{Matthew~R Scott}.}
  \bibinfo{year}{2020}\natexlab{}.
\newblock \showarticletitle{Cross-batch memory for embedding learning}. In
  \bibinfo{booktitle}{\emph{Proceedings of the IEEE/CVF Conference on Computer
  Vision and Pattern Recognition}}. \bibinfo{pages}{6388--6397}.
\newblock


\bibitem[\protect\citeauthoryear{Wei, Wang, Nie, He, Hong, and Chua}{Wei
  et~al\mbox{.}}{2019}]%
        {wei2019mmgcn}
\bibfield{author}{\bibinfo{person}{Yinwei Wei}, \bibinfo{person}{Xiang Wang},
  \bibinfo{person}{Liqiang Nie}, \bibinfo{person}{Xiangnan He},
  \bibinfo{person}{Richang Hong}, {and} \bibinfo{person}{Tat-Seng Chua}.}
  \bibinfo{year}{2019}\natexlab{}.
\newblock \showarticletitle{MMGCN: Multi-modal graph convolution network for
  personalized recommendation of micro-video}. In
  \bibinfo{booktitle}{\emph{Proceedings of the 27th ACM International
  Conference on Multimedia}}. \bibinfo{pages}{1437--1445}.
\newblock


\bibitem[\protect\citeauthoryear{Wu, Manmatha, Smola, and Krahenbuhl}{Wu
  et~al\mbox{.}}{2017}]%
        {wu2017sampling}
\bibfield{author}{\bibinfo{person}{Chao-Yuan Wu}, \bibinfo{person}{R Manmatha},
  \bibinfo{person}{Alexander~J Smola}, {and} \bibinfo{person}{Philipp
  Krahenbuhl}.} \bibinfo{year}{2017}\natexlab{}.
\newblock \showarticletitle{Sampling matters in deep embedding learning}. In
  \bibinfo{booktitle}{\emph{Proceedings of the IEEE International Conference on
  Computer Vision}}. \bibinfo{pages}{2840--2848}.
\newblock


\bibitem[\protect\citeauthoryear{Wu, Petroni, Josifoski, Riedel, and
  Zettlemoyer}{Wu et~al\mbox{.}}{2020}]%
        {wu2020scalable}
\bibfield{author}{\bibinfo{person}{Ledell Wu}, \bibinfo{person}{Fabio Petroni},
  \bibinfo{person}{Martin Josifoski}, \bibinfo{person}{Sebastian Riedel}, {and}
  \bibinfo{person}{Luke Zettlemoyer}.} \bibinfo{year}{2020}\natexlab{}.
\newblock \showarticletitle{Scalable Zero-shot Entity Linking with Dense Entity
  Retrieval}. In \bibinfo{booktitle}{\emph{Proceedings of the 2020 Conference
  on Empirical Methods in Natural Language Processing (EMNLP)}}.
  \bibinfo{pages}{6397--6407}.
\newblock


\bibitem[\protect\citeauthoryear{Xiong, Dai, Callan, Liu, and Power}{Xiong
  et~al\mbox{.}}{2017}]%
        {xiong2017end}
\bibfield{author}{\bibinfo{person}{Chenyan Xiong}, \bibinfo{person}{Zhuyun
  Dai}, \bibinfo{person}{Jamie Callan}, \bibinfo{person}{Zhiyuan Liu}, {and}
  \bibinfo{person}{Russell Power}.} \bibinfo{year}{2017}\natexlab{}.
\newblock \showarticletitle{End-to-end neural ad-hoc ranking with kernel
  pooling}. In \bibinfo{booktitle}{\emph{Proceedings of the 40th International
  ACM SIGIR conference on research and development in information retrieval}}.
  \bibinfo{pages}{55--64}.
\newblock


\bibitem[\protect\citeauthoryear{Xiong, Xiong, Li, Tang, Liu, Bennett, Ahmed,
  and Overwijk}{Xiong et~al\mbox{.}}{2021}]%
        {ance}
\bibfield{author}{\bibinfo{person}{Lee Xiong}, \bibinfo{person}{Chenyan Xiong},
  \bibinfo{person}{Ye Li}, \bibinfo{person}{Kwok-Fung Tang},
  \bibinfo{person}{Jialin Liu}, \bibinfo{person}{Paul~N. Bennett},
  \bibinfo{person}{Junaid Ahmed}, {and} \bibinfo{person}{Arnold Overwijk}.}
  \bibinfo{year}{2021}\natexlab{}.
\newblock \showarticletitle{Approximate Nearest Neighbor Negative Contrastive
  Learning for Dense Text Retrieval}. In \bibinfo{booktitle}{\emph{ICLR}}.
\newblock


\bibitem[\protect\citeauthoryear{Xu, Shen, Liu, and Shen}{Xu
  et~al\mbox{.}}{2018}]%
        {xu2018graphcar}
\bibfield{author}{\bibinfo{person}{Qidi Xu}, \bibinfo{person}{Fumin Shen},
  \bibinfo{person}{Li Liu}, {and} \bibinfo{person}{Heng~Tao Shen}.}
  \bibinfo{year}{2018}\natexlab{}.
\newblock \showarticletitle{Graphcar: Content-aware multimedia recommendation
  with graph autoencoder}. In \bibinfo{booktitle}{\emph{The 41st International
  ACM SIGIR Conference on Research \& Development in Information Retrieval}}.
  \bibinfo{pages}{981--984}.
\newblock


\bibitem[\protect\citeauthoryear{Yang, Yi, Zhiyuan~Cheng, Hong, Li,
  Xiaoming~Wang, Xu, and Chi}{Yang et~al\mbox{.}}{2020b}]%
        {MNS}
\bibfield{author}{\bibinfo{person}{Ji Yang}, \bibinfo{person}{Xinyang Yi},
  \bibinfo{person}{Derek Zhiyuan~Cheng}, \bibinfo{person}{Lichan Hong},
  \bibinfo{person}{Yang Li}, \bibinfo{person}{Simon Xiaoming~Wang},
  \bibinfo{person}{Taibai Xu}, {and} \bibinfo{person}{Ed~H Chi}.}
  \bibinfo{year}{2020}\natexlab{b}.
\newblock \showarticletitle{Mixed Negative Sampling for Learning Two-tower
  Neural Networks in Recommendations}. In \bibinfo{booktitle}{\emph{Companion
  Proceedings of the Web Conference 2020}}. \bibinfo{pages}{441--447}.
\newblock


\bibitem[\protect\citeauthoryear{Yang, Cer, Ahmad, Guo, Law, Constant, Abrego,
  Yuan, Tar, Sung, et~al\mbox{.}}{Yang et~al\mbox{.}}{2020a}]%
        {yang2020multilingual}
\bibfield{author}{\bibinfo{person}{Yinfei Yang}, \bibinfo{person}{Daniel Cer},
  \bibinfo{person}{Amin Ahmad}, \bibinfo{person}{Mandy Guo},
  \bibinfo{person}{Jax Law}, \bibinfo{person}{Noah Constant},
  \bibinfo{person}{Gustavo~Hernandez Abrego}, \bibinfo{person}{Steve Yuan},
  \bibinfo{person}{Chris Tar}, \bibinfo{person}{Yun-hsuan Sung},
  {et~al\mbox{.}}} \bibinfo{year}{2020}\natexlab{a}.
\newblock \showarticletitle{Multilingual Universal Sentence Encoder for
  Semantic Retrieval}. In \bibinfo{booktitle}{\emph{Proceedings of the 58th
  Annual Meeting of the Association for Computational Linguistics: System
  Demonstrations}}. \bibinfo{pages}{87--94}.
\newblock


\bibitem[\protect\citeauthoryear{Yi, Yang, Hong, Cheng, Heldt, Kumthekar, Zhao,
  Wei, and Chi}{Yi et~al\mbox{.}}{2019}]%
        {corr-sfx}
\bibfield{author}{\bibinfo{person}{Xinyang Yi}, \bibinfo{person}{Ji Yang},
  \bibinfo{person}{Lichan Hong}, \bibinfo{person}{Derek~Zhiyuan Cheng},
  \bibinfo{person}{Lukasz Heldt}, \bibinfo{person}{Aditee Kumthekar},
  \bibinfo{person}{Zhe Zhao}, \bibinfo{person}{Li Wei}, {and}
  \bibinfo{person}{Ed Chi}.} \bibinfo{year}{2019}\natexlab{}.
\newblock \showarticletitle{Sampling-bias-corrected neural modeling for large
  corpus item recommendations}. In \bibinfo{booktitle}{\emph{Proceedings of the
  13th ACM Conference on Recommender Systems}}. \bibinfo{pages}{269--277}.
\newblock


\bibitem[\protect\citeauthoryear{Yih, Toutanova, Platt, and Meek}{Yih
  et~al\mbox{.}}{2011}]%
        {yih2011learning}
\bibfield{author}{\bibinfo{person}{Wen-tau Yih}, \bibinfo{person}{Kristina
  Toutanova}, \bibinfo{person}{John~C Platt}, {and}
  \bibinfo{person}{Christopher Meek}.} \bibinfo{year}{2011}\natexlab{}.
\newblock \showarticletitle{Learning discriminative projections for text
  similarity measures}. In \bibinfo{booktitle}{\emph{Proceedings of the
  fifteenth conference on computational natural language learning}}.
  \bibinfo{pages}{247--256}.
\newblock


\bibitem[\protect\citeauthoryear{Zhou, Ma, Zhang, Zhou, and Yang}{Zhou
  et~al\mbox{.}}{2020}]%
        {zhou2020contrastive}
\bibfield{author}{\bibinfo{person}{Chang Zhou}, \bibinfo{person}{Jianxin Ma},
  \bibinfo{person}{Jianwei Zhang}, \bibinfo{person}{Jingren Zhou}, {and}
  \bibinfo{person}{Hongxia Yang}.} \bibinfo{year}{2020}\natexlab{}.
\newblock \showarticletitle{Contrastive Learning for Debiased Candidate
  Generation in Large-Scale Recommender Systems}.
\newblock \bibinfo{journal}{\emph{arXiv preprint cs.IR/2005.12964}}
  (\bibinfo{year}{2020}).
\newblock


\end{thebibliography}
\end{document}